\documentclass[onecolumn,letterpaper]{IEEEtran}

\usepackage{amsthm}
\usepackage{amssymb}
\usepackage{amsfonts}
\usepackage{bbm}
\usepackage[utf8]{inputenc} 
\usepackage[T1]{fontenc}
\usepackage{url}
\usepackage{ifthen}
\usepackage{cite}
\usepackage{todonotes}
\usepackage{xcolor}
\usepackage[cmex10]{amsmath} % Use the [cmex10] option to ensure complicance
\input{Definition.def}

\begin{document}
\title{Few-Shot Channel-Agnostic Analog Coding: \\A Near-Optimal Scheme} 

% %%% Single author, or several authors with same affiliation:
% \author{%
%  \IEEEauthorblockN{Andrew R.~Barron}
%  \IEEEauthorblockA{Department of Statistics and Data Science\\
%                    Yale University\\
%                    New Haven, CT, USA\\
%                    Email: andrew.barron@yale.edu}
% }

%%% Several authors with up to three affiliations:
\author{%
  \IEEEauthorblockN{Mohammad Ali Maddah-Ali and Soheil Mohajer}
  \IEEEauthorblockA{Department of Electrical and Computer Engineering, University of Minnesota, Minneapolis, MN, USA\\
                    Email: maddah@umn.edu, soheil@umn.edu}

}

%%% Many authors with many affiliations:
% \author{%
%   \IEEEauthorblockN{Andrew R.~Barron\IEEEauthorrefmark{1},
%                     Claude E.~Shannon\IEEEauthorrefmark{2},
%                     David Slepian\IEEEauthorrefmark{2},
%                     and Jacob Ziv\IEEEauthorrefmark{2}\IEEEauthorrefmark{3}}
%   \IEEEauthorblockA{\IEEEauthorrefmark{1}%
%                    Department of Statistics and Data Science, Yale University, New Haven, CT, USA,
%                     andrew.barron@yale.edu}
%   \IEEEauthorblockA{\IEEEauthorrefmark{2}%
%                     Bell Telephone Laboratories, Inc.,
%                     Murray Hill, NJ, USA,
%                     \{csh,dsl,jz\}@bell-labs.com}
%   \IEEEauthorblockA{\IEEEauthorrefmark{3}%
%                     Department of Electrical Engineering, Technion---Institute of Technology, Haifa, Israel,
%                     jz@ee.technion.ac.il}
% }

\maketitle

%%%%%%
%% Abstract: 
%% If your paper is eligible for the student paper award, please add
%% the comment "THIS PAPER IS ELIGIBLE FOR THE STUDENT PAPER
%% AWARD." as a first line in the abstract. 
%% For the final version of the accepted paper, please do not forget
%% to remove this comment!
%%

\begin{abstract}
In this paper, we investigate the problem of transmitting an analog source to a destination over $N$ uses of an additive-white-Gaussian-noise (AWGN) channel, where $N$ is very small (in the order of 10 or even less). The proposed coding scheme is based on representing the source symbol using a novel progressive expansion technique, partitioning the digits of expansion into $N$ ordered sets, and finally mapping the symbols in each set to a real number by applying the reverse progressive expansion. In the last step, we introduce some gaps between the signal levels to prevent the carry-over of the additive noise from propagation to other levels. This shields the most significant levels of the signal from an additive noise, hitting the signal at a less significant level. The parameters of the progressive expansion and the shielding procedure  are opportunistically independent of the $\SNR$ so that the proposed scheme achieves a distortion $D$, where $-
\log(D)$ is within $O(\log\log(\SNR))$ of the optimal performance for all values of $\SNR$, leading to a channel-agnostic scheme. 
\end{abstract}

\section{Introduction}
Consider a joint source-channel coding problem, where the objective is to transmit some sequence of identically and independently distributed (i.i.d.) source $\{U_m\}_{m=1}^M$ to a destination using $N$-uses of a memory-less channel. The destination aims to recover $\{U_m\}_{m=1}^M$ with minimum distortion for some distortion measure function. For this problem, we can consider three different scenarios depending on $M$ and $N$, namely, long-length codes, short-length codes, and few-shot codes.  
  
In \emph{long-length codes}, where $N \rightarrow \infty$, the optimal approach is based on the separation of source coding and channel coding. That is, the transmitter first employs a lossy source coding technique to quantize the information symbols with some distortion, and then it uses channel coding to send the (index of) quantized symbols to the receiver. Optimality of separation implies that the decoded symbols satisfy the average distortion 
%this leads to a vanishing probability of error 
as long as the rate of source coding is less than the channel capacity. This scheme achieves the minimum distortion for any channel. 

In \emph{short-length codes}, $N$ is in the order of a few hundred. %For joint source-coding channel problems in this regime, 
It is shown that in this regime, the separation of source coding and channel coding results in a significant loss of performance compared to the joint schemes. In addition, the optimum scheme satisfies 
${NC - MR(D) \approx \sqrt{NV+M\mathcal{V}(D)
}Q^{-1}(\epsilon)}$, where  $\epsilon$  denotes the probability error, $D$ is the distortion, $V$ represents the channel dispersion, $Q^{-1}(\cdot)$ is the inverse of the complementary cumulative distribution function of a normal random variable, and $R(D)$ and $\mathcal{V}(D)$ are the source rate-distortion and rate-dispersion functions~\cite{kostina2013lossy, polyanskiy2010channel}.

\emph{Few-shot communication}
 is applicable when the parameter $N$ is on the order of $10$ or even less. Such a regime appears in several critical, delay-sensitive applications. For example, this scenario is raised in communication systems to minimize delay in reporting time-sensitive channel state information (CSI) from the receiver to the transmitter~\cite{caire2010multiuser}. Up-to-date CSI at the transmitter is particularly important for interference management in dense wireless networks, e.g., small cells and cooperative multi-antenna communication. 
Another application of few-shot communication appears is reducing delay in reporting real-time system outputs to the controller~\cite{tatikonda2004control}. This is particularly crucial in situations where the controller and the plant are not collocated, as seen in applications like remote surgery. A significant use-case is found in massive sensor networks, where each sensor is assigned a few narrow-band time slots to report its low-rate measurements.

Few-shot codes have been investigated in the context of analog coding~\cite{shannon1949communication, cagnazzo2015shannon, timor1970design, chen1998analog, vaishampayan2003curves, skoglund2002design, skoglund2006hybrid, mittal2002hybrid, santhi2006analog, taherzadeh2012single}. However, this scenario is of particular interest in cases where the channel state information, especially the signal-to-noise ratio (SNR), is unknown at the transmitter. In the most common use-cases of few-shot codes, the delay involved in sending pilot signals, estimating the channel, and disseminating channel information is not tolerable, or the necessary resources are unavailable. Therefore, it is crucial to develop a \emph{robust} or \emph{channel-agnostic} scheme that performs optimally across all values of SNR.

The problem of designing robust schemes has been studied in~\cite{santhi2006analog, taherzadeh2012single}.
%are channel agnostic. 
Let $\SDR=\frac{\E[(U\!-\!\E[U])^2}{D}$, and $\SNR$ be the signal-to-noise ratio. Then, the information-theoretic bound suggests that ${-\log(\SDR) \lessapprox N \log(\SNR)}$. The scheme proposed in~\cite{santhi2006analog} achieves $\frac{-\log(\SDR)} {N \log(\SNR)}\! <\! 1$, and thus  $-\log(\SDR) -N \log(\SNR)$ is in the order of $O(\log(\SNR))$. This result is significantly improved by a scheme proposed in~\cite{taherzadeh2012single},  
%offers significant improvement compared to that of ~\cite{santhi2006analog} and 
that achieves $-\log(\SDR) -N \log(\SNR)=O(\sqrt{\log(\SNR)})$.

In this paper, we introduce a new few-shot source-channel coding scheme that, for $M=1$ source symbol, achieves
${-\log(\SDR) -N \log(\SNR)=O(\log\log(\SNR))}$. 
This scheme is built on proposing a new expansion, named \emph{progressive expansion}, of the source $U$. Recall that in the conventional fixed-base expansion, such as binary expansion, digits are derived by recursively multiplying the residual by the base and computing the quotient. A similar process is executed to derive the digits in progressive expansion, except that the \emph{base} of the expansion grows over iterations, following a particular pattern. Then, the digits of the source expansion are partitioned into $N$ ordered sets in a round-robin scheduling manner. The digits in each set are then mapped to an analog signal by applying a reverse-progressive expansion over each set. In each level of reverse-progressive expansion, we leave the first and last alphabets unused to shield the earlier levels from the carry-over of the additive noise, hitting the signal at a higher level. The parameters of the progressive expansion and also the number of unused alphabets are judiciously designed such that the proposed scheme achieves near-optimum distortion for all values of $\SNR$, even though the achievable scheme is channel-agnostic and does not depend on $\SNR$. 
%
%
% In the progressive expansion, however,  iterations are grouped into blocks of $N$ iterations. Within each block, the base remains constant. However, starting from the first block with base 2, the base for each subsequent block increases by one compared to the preceding block. 
%Then in the proposed scheme, from each block, one digit is assigned to each channel, and the transmitted symbol is calculated by reverse-progressive expansion with a block length of one. During this process, the first and last alphabet in each layer of expansion are left unused. This precaution is taken to shield the signal from \soh{(positive or negative)} carry-overs of added noise, ensuring that the corruption does not \soh{propagate to all other digits.} affect all levels. The proposed scheme achieves a balanced trade-off between the length of the first block in the expansion affected by additive noise and the number of unused alphabets up to that block, \emph{no matter at which level in the expansion, noise hits the signal}. 
As a result, the proposed scheme remains very close to the optimum for all values of $\SNR$.

\textbf{Notation.} For $N \in \mathbb{N}$, we define ${[N] := \{1, \ldots,N\}}$. For an interval $\cI=[a,b)$ we denote its length by $|\cI|=b-a$. For a subset of $\cJ\subseteq \mathbb{R}$ that is a union of a collection of mutually disjoint intervals, i.e., $\cJ =\bigcup_{i} \cI$ we have ${|\cJ|=\sum_{i} |\cI|}$. Finally, $\log(\cdot)$ denotes logarithm with in base~$2$. 

\section{Problem Formulation}
We consider a \emph{few-shot} lossy joint source-channel coding problem, where the source $U_1, U_2, \ldots, U_M$, for some ${M\in \mathbb{N}}$, is a sequence of independently and identically distributed random variables, drawn according to some probability density function (PDF) $f_U(.)$. 
The objective is to communicate this source through ${N\in \mathbb{N}}$ uses of a memoryless AWGN channel. Motivated by delay-limited applications, we focus on the settings where $N$ is very small.

An $(M,N)$ few-shot lossy joint source-channel coding scheme consists of $N$ encoding function and a decoding function. 
For each $n=1,2,\ldots,N$, the encoding function %$\phi_{m}$, $n=1,2,\ldots,N$
\begin{align*}
    \phi_{n}: \mathbb{R}^M\rightarrow \mathbb{R}
\end{align*}
maps the source sequence $U_1, U_2, \ldots, U_M$ to a transmit symbol $\x{n}=\phi_{n} (U_1, U_2, \ldots, U_M)$. For each $n\in[N]$, the coded symbol $\x{n}$ is transmitted through an AWGN channel, 
\begin{align*}
    \y{n} = \x{n} + \z{n}, 
\end{align*}
where  $(\z{1},\ldots, \z{N})$ is an i.i.d. Gaussian sequence with zero mean and  variance~$\sigma^2$. The decoding function
\begin{align*}
    \psi:\mathbb{R}^N \rightarrow \mathbb{R}^M
\end{align*}
maps the received symbols $\y{1},\ldots, \y{N}$ to an estimate for the source sequence, that is, 
\begin{align*}
    (\hU_1, \ldots, \hU_M)= \psi( \y{1},\ldots, \y{N}).
\end{align*}
The distortion of this estimation is defined as 
\begin{align*}
    \frac{1}{M}\sum\nolimits_{m=1}^M |\hU_m-U_m|^2. 
\end{align*}
For  $P,D \in \R_+$, a tuple $(M,N,P,D)$ is said to be achievable if there exists an $(M,N)$ few-shot joint source-channel coding scheme,  such that the following constraints are satisfied:
\begin{itemize}
    \item Power constraint:
$\E\left[\frac{1}{N}\sum\nolimits_{n=1}^N |\x{n}|^2 \right]  \leq P$.

\item  Distortion constraint:
$\E\left[\frac{1}{M}\sum\nolimits_{m=1}^M |\hU_m-U_m|^2 \right] \leq D$.

\end{itemize}
\begin{remark}
    It is important to note that, here, we do not assume that $M$ and $N$ go to infinity. Motivated by delay-limited scenarios, we assume that $M$ and $N$ are small integers (e.g., $M=1$ and $N=3$). 
\end{remark}

We call an achievable scheme \emph{channel agnostic} if the encoding functions do not depend on $P$ and $\sigma$. Then, we define 
\begin{align*}
D^*(M,N,P):= \inf\{D: & (M,N,P,D) \textrm{ is achievable by a} \\ &\textrm{channel-agnostic scheme}\}.
\end{align*}

The objective of this paper is to characterize $D^*(M,N,P)$. In particular, for simplicity, we focus on $M=1$ source symbol that admits a uniform\footnote{The proposed scheme can be easily generalized to any source distribution with bounded range. For unbounded distributions (e.g., Gaussian), our scheme can be still adopted. However, the large deviation error probability should be incorporated into the final result. Also, a similar idea can be used when we have multiple source samples to be communicated. } distribution over $[-\frac{1}{2}, \frac{1}{2}]$. 

\subsection{The Fundamental Limits}
Following Shannon's theorem on the optimality of separation of source-channel coding~\cite{csiszar2011information}, we can show that for any $(M,N,P,D)$-achievable few-shot joint source-channel coding, we have 
\begin{align}\label{eq:shannon}
    M\Big(h(U)-\frac{1}{2}\log&(2\pi e D)\Big) \leq M R_U(D)\nonumber\\
    &\leq N C_Z(P)=\frac{N}{2} \log\Big(1+ \frac{P}{\sigma^2}\Big),
\end{align}
where $h(U)$ denotes the differential entropy of the source~$U$. Let $\SDR$ denote the signal to distortion ratio, defined as ${\SDR=\frac{\sigma_U^2}{D}}$ where $\sigma_U^2=\E[(U-\E[U])^2$ is the signal variance, and $\SNR$ denote the signal to noise ratio, defined as $\SNR=\frac{P}{\sigma}$. Then, from~\eqref{eq:shannon}, we have
\begin{align}\label{eq:shannon:2}
\SDR^{M} \leq c (1+\SNR)^N,
\end{align}
where $c=(2\pi e \sigma_U^2 2^{-2h(U)} )^M$. For $U\sim \mathsf{Unif}([-1/2,1/2])$, we have $h(U)=0$ and $\sigma_U^2 = 1/12$, leading to $c=(\pi e/6)^M$. 

\section{The Main Results}
Shannon's theorem guarantees that the upper bound in~\eqref{eq:shannon:2} can be asymptotically achieved by long-length code. However, for few-shot codes, the best achievable $\SDR$ is unknown. 

The main contribution of this paper is as follows:  In Section~\ref{sec:scheme} we propose a few-shot joint source-channel coding scheme (for $M=1$ source symbol).  Then, we study the performance of the proposed scheme in Section~\ref{sec:analysis}, and derive an achievable bound for $\SDR$. The following theorem states the main result of the paper. We refer to Section~\ref{sec:analysis} for the proof of Theorem~\ref{thm:main}.
\begin{theorem}\label{thm:main}
 For $M=1$ source symbol distributed as ${U\sim \mathsf{Unif}([-1/2,1/2])}$, the average distortion 
\begin{align*}
    D = c_1 \frac{1}{\SNR^N} \left(\log \SNR\right)^{10N} + c_2 \frac{1}{\SNR^N} + c_3 
\end{align*}
 is achievable for some constants $(c_1,c_2)$ that do not depend on $\SNR$ or $D$, (but may depend on $N$), and $c_3$ is order-wise smaller than $1/\SNR^N$.
 \end{theorem}

Focusing on the high $\SNR$, we have the following characterization for $\SDR^*=\frac{\sigma_U^2}{D^\star}$. 
 \begin{corollary}\label{cor:SDR}
 For $M=1$ source ${U\sim \mathsf{Unif}([-1/2,1/2])}$, the optimum $\SDR$ satisfies
    \begin{align}
    N\log(\SNR)&-10N \log\log(\SNR) + o(\log\log(\SNR)) \nonumber\\
    &\leq
        \log(\SDR^*) \leq N\log(1+\SNR) \nonumber.
    \end{align}
\end{corollary} 
We refer to Section~\ref{sec:analysis} for the proof of Corollary~\ref{cor:SDR}.

\section{An Achievable Scheme for $M=1$}\label{sec:scheme}
\subsection{Preliminaries: A Progressive Expansion}
Let $x \in [0,1]$. For an integer $S$, we define the $S$-progressive expansion of $x$ as 
\begin{align}
x=\sum_{k=1}^{\infty}\sum_{s=1}^{S}  
    \frac{x_{ks}}{(k!)^S (k+1)^{s}},
\end{align}
where $x_{ {ks}} \in \{0, 1, \ldots, k\}$ for every $s\in \{1,\dots, S\}$.

\begin{remark}[{\bf Fixed-Based vs. progressive expansion}]\label{rem:pr-bin}
Recall that in the conventional fixed-base expansion, such as binary expansion, digits are derived by recursively multiplying the residual by the base and computing the quotient. In $S$-progressive expansion, however, iterations are grouped into blocks of $S$ iterations. Within each block, the base remains constant. However, starting from the first block with base 2, the base for each subsequent block increases by one compared to the preceding block.  Note that when $S\to\infty$, then the $S$-progressive expansion reduces to $x=\sum_{s=1}^\infty \frac{x_{1s}}{2^s}$, which is the binary expansion of $x$. 
\end{remark}

\begin{figure}
    \centering
    \includegraphics[width=0.48\textwidth]{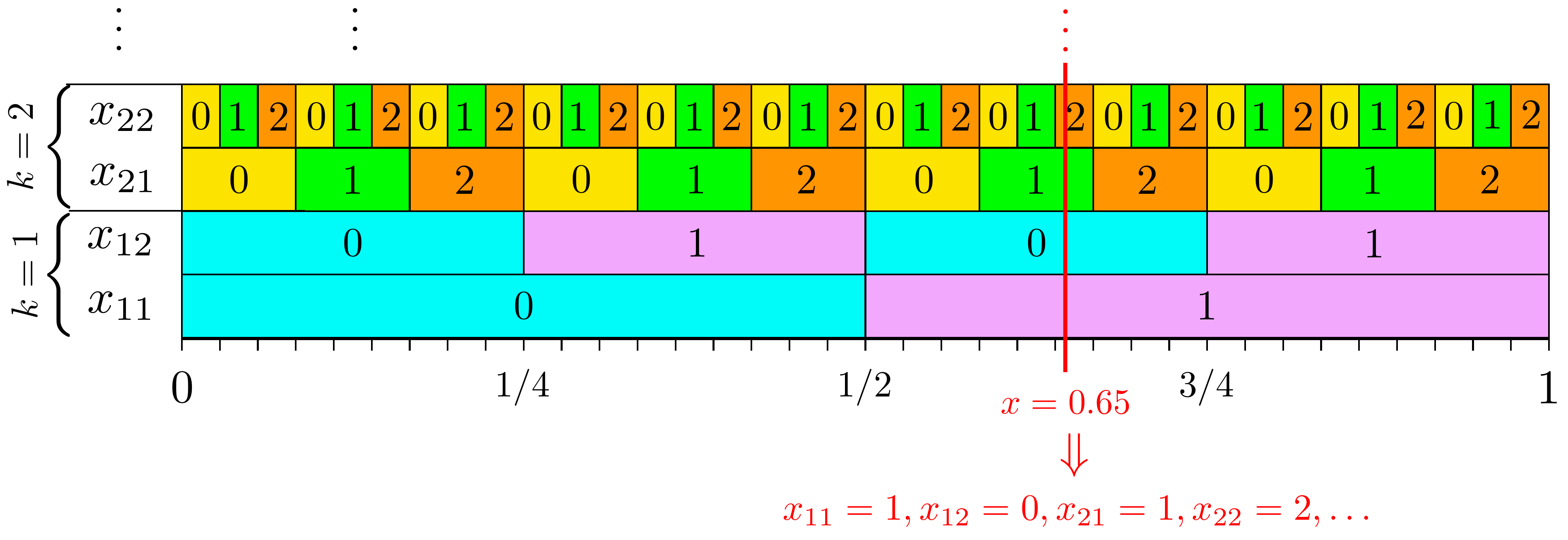}
    \caption{An example of progressive expansion with $S=2$. }
    \label{fig:expansion}
\end{figure}
An example of $2$-progressive expansion ($S=2$) is shown in Figure~\ref{fig:expansion}. As mentioned in Remark~\ref{rem:pr-bin}, the first two symbols $(x_{12},x_{22})$ are obtained similar to a binary expansion. Then, $(x_{21},x_{22})$ in each interval can be found similar to ternary expansion. The following symbols, $x_{31}, x_{32}, x_{41}, \dots, $ can be found in a similar manner.

We will use the following notation throughout the paper. 
\begin{defi}\label{des:order}
For two pairs $(k,s),(\ell,t)\in \mathbb{N}\times \{1,\dots, S\}$ we say $(k,s)\succ (\ell,t)$ if either $k>\ell$, or $k=\ell$ and ${s>t}$. Similarly, $(k,s)\succcurlyeq (\ell,t)$ means wither $(k,s)=(\ell,t)$ or $(k,s)\succ (\ell,t)$. 
\end{defi}

In the following, we provide a list of lemmas that state the properties of the $S$-progressive expansion. These properties will be used later for the proof of the main result. The proofs of the lemmas can be found in Appendix~\ref{expansionProof}.
\begin{lemma}\label{lm:unique}
The $S$-expansion of any $x \in [0,1]$ is unique.\footnote{Recall that for binary expansion we have $(0.1000\ldots)_2= (0.0111\ldots)_2$. A similar boundary scenario also occurs for progressive expansion. To avoid confusion, we stick to the finite (i.e., there exists some $(\ell,t)$ such that  $x_{ks}=0$ for every $(k,s)\succ (\ell,t)$) expansion of $x$ if that exists. }
\end{lemma}

\begin{lemma}\label{lm:sum-fact}
    For any $\ell\in \N$ we have
    \[
\sum_{k=\ell}^{\infty} \frac{k}{(k+1)!}=\frac{1}{\ell!}.
    \]
\end{lemma}

\begin{lemma}\label{lm:exp-small}
    If $x<\frac{1}{(\ell !)^S(\ell+1)^t}$ then for the $S$-progressive expansion we have 
$x_{ks}=0$ for every $(k,s)\preccurlyeq (\ell,t)$.
\end{lemma}

\begin{lemma}
\label{lm:exp-prop}
Let $X$ be a variable chosen uniformly at random from $[0,1]$ with $S$-progressive expansion 
\begin{align}
X=\sum_{k=1}^{\infty}\sum_{s=1}^{S}   \frac{X_{ks}}{(k!)^S (k+1)^s},
\end{align}
then 
\begin{itemize}
\item $X_{ks}$ has a uniform distribution in $\{0, 1, \ldots, k\} $.

\item The sequence of random variables $\{ \{ X_{ks} \}_{s=1}^{S} \}_{k=1}^{\infty}$ are pairwise independent.
\end{itemize}
\end{lemma}

\begin{lemma}\label{lm:fact-inverse-bound}
    If $\omega\geq 4$ and $n\geq 6$ satisfy $ \omega \leq n!$, then we have $n-1\geq \frac{\log \omega}{\log \log \omega}$. 
\end{lemma}

\subsection{An Achievable Scheme}

In this section, we focus on $M=1$ source symbol and $N$ channel uses. We also assume that $U\sim \mathsf{Unif}([-\frac{1}{2},\frac{1}{2}])$. Lastly, without loss of generality, we assume unit power constraint, i.e., $\E[|\x{n}|^2]\leq P=1$.  \\[-1mm]

    \noindent {$\bullet$ \bf Encoding:} We start with the $N$-progressive expansion of ${U+\frac{1}{2}}$, that is  
    \begin{align}\label{eq:source-exp}
    U+\frac{1}{2}&= \sum_{k=1}^{\infty}\sum_{n=1}^{N}
    \frac{U_{kn}}{(k!)^N (k+1)^n},
\end{align}
where $U_{kn} \in \{0, 1, \ldots,  k\} $.  Lemma~\ref{lm:unique} implies that this expansion is unique. In addition, from Lemma~\ref{lm:exp-prop} we know that  $U_{kn}$ is uniformly distributed over  $\{0, 1, \ldots,  k\}$ and $U_{ks}$, $k=1,2,\ldots$, and $s=1, \ldots, S$ are pairwise independent.

For each channel use $n\in[N]$, we generate 
\begin{align}\label{eq:channel-input}
    \tx{n}= 3!\sum_{k=1}^{\infty}  
    \frac{U_{kn}+1}{(k+3)!} - \frac{1}{2}
\end{align}

The following lemma provides the properties of $\tx{n}$. We refer to the Appendix for the proof of the lemma. 
\begin{lemma}\label{lm:input-properties}
    For $\tx{n}$ defined in~\eqref{eq:channel-input}, we have 
    ${\tx{n} \in [-\alpha, \alpha]}$ where $\alpha = \frac{\sqrt{3}}{(16.5-6e)}$. Moreover, we have $\E[\tx{n}]=0$ and $\E[|\tx{n}|^2] =  0.0173814$.
\end{lemma}
Then, the channel input for the $n$th channel is given by 
\begin{align*}
 \x{n}= \gamma\tx{n},
\end{align*}
with $\gamma :=7.585= \frac{1}{\sqrt{\E[|\tx{n}|^2]}}$, so that the power constraint is satisfied, i.e., $\E[X^2(n)]=1$  for every $n\in [N]$.

\begin{remark} [{\bf The Intuition Behind the Achievable Scheme}]
    Note that from each block in~\eqref{eq:source-exp} i.e., $\{U_{k1}, U_{k2}, \dots, U_{kN}\}$, one digit is assigned to each channel, and the transmitted symbol is calculated by reverse-progressive expansion with a block length of one ($S=1$), as in~\eqref{eq:channel-input}. During this process, the first and last alphabet in each layer of expansion are left unused. More precisely, in the progressive expansion, the coefficient of $1/(k+3)!$ can be any number in $\{0,1,\dots, k+2\}$. However, $U_{kn1}+1$ can only take values in $\{1,\dots, k+1\}$. That is, $0$ and $k+2$ are eliminated to avoid negative and positive noise (respectively) at level $1/(k+4)!$ to reach level $1/(k+2)!$. This precaution is taken to shield the signal from (positive or negative) carry-overs of added noise, ensuring that the corruption does not propagate to all other digits. The proposed scheme achieves a balanced trade-off between the length of the first block in the expansion affected by additive noise and the number of unused alphabets up to that block, \emph{no matter at which level in the expansion, noise hits the signal}. 
As a result, the proposed scheme remains very close to the optimum for all values of $\SNR$.
\end{remark}

\begin{figure}
    \centering
    \includegraphics[width=0.35\textwidth]{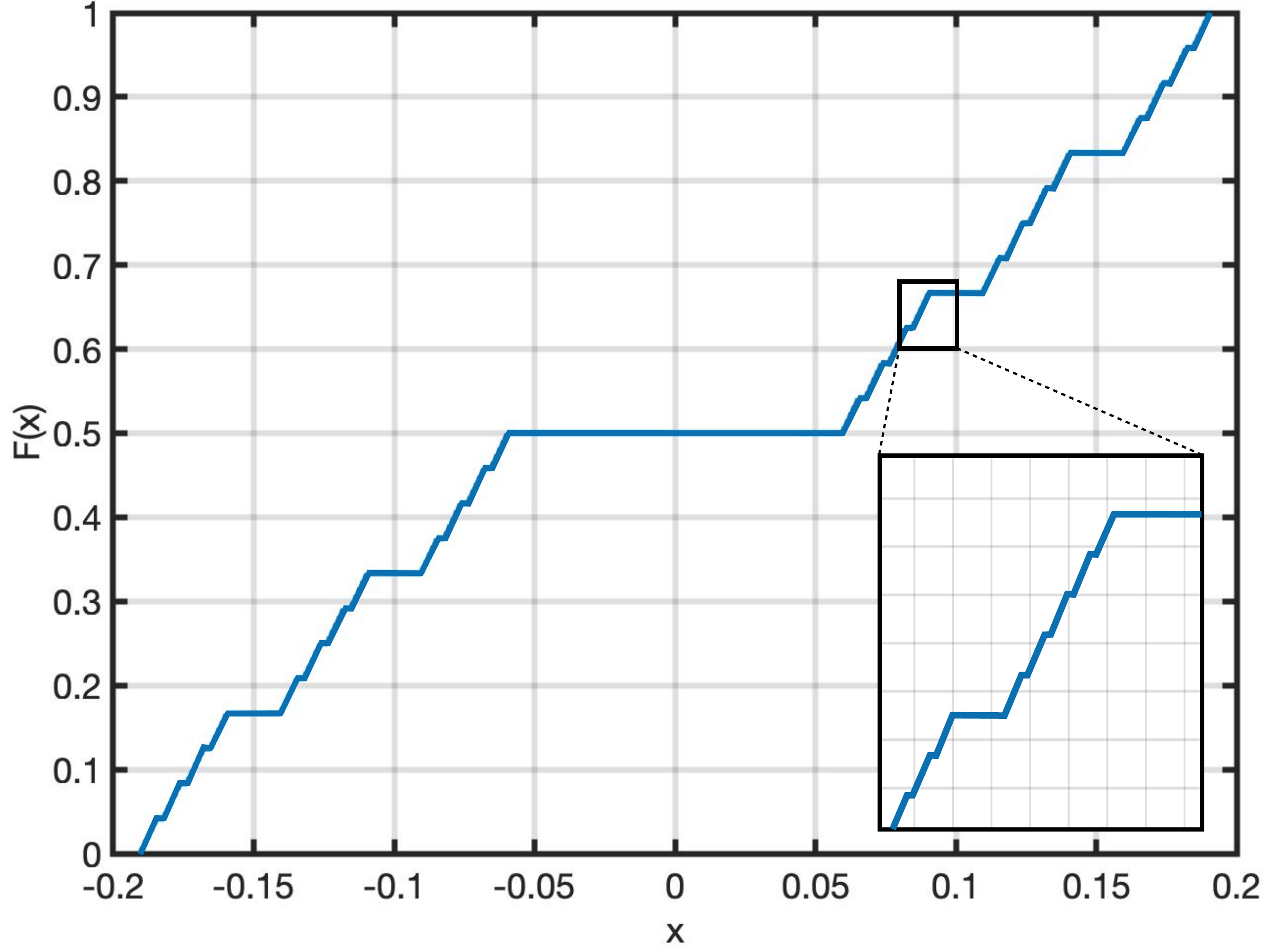}
    \vspace{-5pt}
    \caption{The cumulative distribution function of $\tx{n}$ generated from uniformly distributed $U$.}
    \label{fig:enter-label}
    \vspace{-5pt}
\end{figure}

Figure~\ref{fig:enter-label} shows the cumulative distribution function (CDF) of $\tx{n}$ generated from uniformly distributed $U$, according to~\eqref{eq:channel-input}. It can be seen that $\tx{n}$ does \emph{not} admit a uniform distribution. Moreover, there is a semi-fractal behavior in the CDF. For instance, for no value of $U$, $\tx{n}$ admits a value in $[-0.0597,0.0597]$. Similarly, there are non-occurring intervals with smaller lengths on the negative and positive sides.

\begin{figure}
    \centering
\includegraphics[width=.488\textwidth]{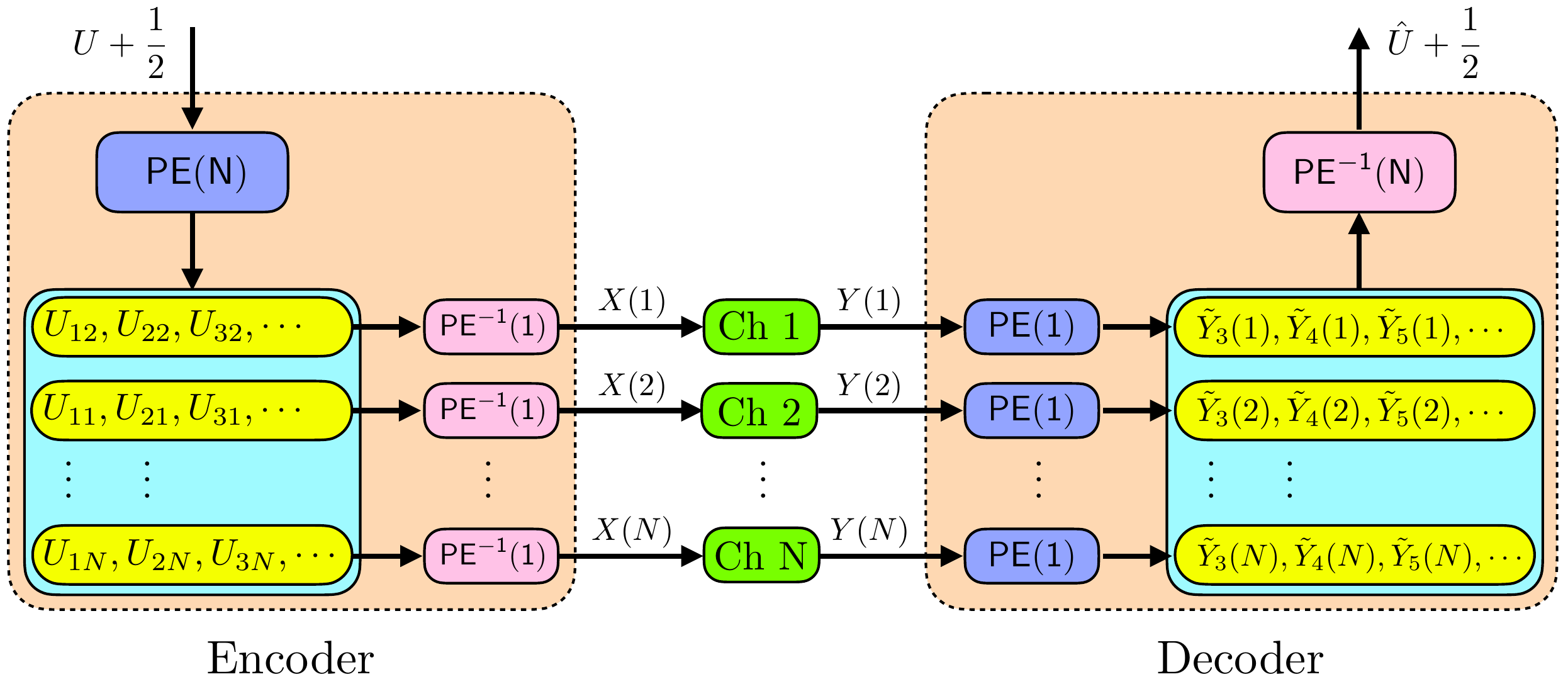}
    \caption{The block diagram of the encoder and decoder. The $\mathsf{PE}(S)$ blocks represent the $S$-progressive expansion, while $\mathsf{PE}^{-1}(S)$ shows its inverse (progressive expansion to decimal conversion). }
    \label{fig:Diagram}
    \vspace{-5pt}
\end{figure}

\noindent $\bullet$ {\bf Decoding: }
Let
\begin{align}
    \ty{n}=\min \{ \max\{-\frac{1}{2}, \frac{1}{\gamma}\y{n}\}, \frac{1}{2}\}.
\end{align}
Then, $0\leq \frac{1}{3!}(\frac{1}{\gamma}\ty{n}+\frac{1}{2}) \leq \frac{1}{3!}$. Therefore, Lemma~\ref{lm:exp-small} implies that it can be written as 
\begin{align}
    \frac{1}{3!}\left(\ty{n}+\frac{1}{2}\right)= 
    \sum_{k=3}^{\infty}  
    \frac{\ty[k]{n}}{(k+1)!}=
    \sum_{k=1}^{\infty}  
    \frac{\tilde{V}_{k}(n)}{(k+3)!}.
\end{align}
where $\tilde{V}_{k}(n):=\ty[k+2]{n}\in \{0,1, \ldots, k\!+\!2\}$. 
For every ${k\in \mathbb{N}}$ and every channel $n\in [N]$, we define ${\hat{U}_{kn}:=\tilde{V}_{k}(n)\!-\!1}$. Then, combining all the channel outputs, we compute and declare
\begin{align}
    \hat{U}= \sum_{k=1}^{\infty}\sum_{n=1}^{N}  
    \frac{\hu[k]{n}}{(k!)^N (k+1)^n}-\frac{1}{2},
\end{align}
as the estimate of $U$. 

Figure~\ref{fig:Diagram} illustrates the block diagram of the proposed encoder and decoder.

\section{Performance Analysis}
\label{sec:analysis}
Without loss of generality, let us assume that $P=1$. In addition, here we focus on the cases where $\sigma \leq \frac{6\gamma}{4}  \approx 11.37$.
For a given $\sigma$, consider $\ell \in \mathbb{N}$ that satisfies 
\begin{align}\label{eq:ell-range}
    (\ell+3)! \leq \frac{6\gamma}{\sigma} < (\ell+4)!.
\end{align}
Based on the regions given in~\eqref{eq:ell-range}, we can introduce two regimes for $Z_n$. 
\begin{prop}\label{prop:bad}
\begin{align*}
    \P\left(\frac{|Z_n|}{6\gamma} \geq \frac{1}{(\ell+1)!}\right)  \leq \exp\left(-\frac{(\log (6\gamma/\sigma))^2}{2(\log \log (6\gamma/\sigma))^2}\right).
\end{align*} 
\end{prop}

\begin{proof}
\begin{align}
     \P&\left(\frac{|Z_n|}{6\gamma} \geq \frac{1}{(\ell+2)!}\right)  = \P \left(
     \frac{|Z_n|}{\sigma}
     \geq \frac{\gamma 
     \frac{6}{(\ell+2)!}}{\sigma}  \right)\nonumber\\
     & \stackrel{\rm (a)}{\leq} \P \left(
     \frac{|Z_n|}{\sigma}
     \geq 
     \frac{(\ell+3)!}{(\ell+2)!}    \right)  = \P \left(
     \frac{|Z_n|}{\sigma}
     \geq  \ell+3
      \right)\nonumber\\
      & =2Q\left( \ell+3 \right)  \stackrel{\rm (b)}{\leq}  \exp(- (\ell+3)^2/2)\nonumber\\
      &\stackrel{\rm (c)}{\leq}
      \exp\left(-\frac{(\log (6\gamma/\sigma))^2}{2(\log \log (6\gamma/\sigma))^2}\right), 
\end{align}
where ${\rm (a)}$ follows from~\eqref{eq:ell-range},  in ${\rm (b)}$ we used the fact that  $Q{(u) \leq \frac{1}{2}\exp(-u^2/2)}$, and ${\rm (c)}$ follows from Lemma~\ref{lm:fact-inverse-bound} for $\omega=6\gamma/\sigma$ and $n=\ell+4$. 
\end{proof}

Next, we show that in the second noise regime, the first $\ell-2$ of the decoded symbols match those of the source symbols.
\begin{prop}\label{prop:each-good}
If $\frac{|Z_n|}{6\gamma} <  \frac{1}{(\ell+2)!}$, then we have 
    $\tilde{U}_{kn}=U_{kn}$, for  ${k=1,\ldots, \ell-2}$.   
\end{prop}
\begin{proof}
Let $\tz{n}:=\frac{|\z{n}|}{3!\gamma}$ and  $\tz{n}=    \sum_{k=1}^{\infty} \frac{\tz[k]{n}}{(k+1)!}$ with $\tz[k]{n} \in \{ 0, \ldots, k\}$
be the $1$-expansion (with $S=1$) of $\tz{n}$. 
    We note that ${\tz{n} <\frac{1}{(\ell+2)!}}$. Therefore, Lemma~\ref{lm:exp-small} implies that $\tz[k]{n}=0$ for every $k< \ell+2$. Hence, we have 
\begin{align}
    \frac{|\z{n}|}{3!\gamma} &= \tz{n}   = \sum_{k=\ell+2}^{\infty} \frac{\tz[k]{n}}{(k+1)!} = \sum_{k=\ell}^{\infty} \frac{\z[k]{n}}{(k+3)!}, 
\end{align}
where $\z[k]{n}:=\tz[k+2]{n} \in \{ 0, \ldots, k+2\}$. Thus
\begin{align}\label{eq:output-exp}
    &\sum_{k=1}^{\infty}  
    \frac{  \ty[k]{n}}{(k+3)!} = \frac{1}{3!}\left(\ty{n}+\frac{1}{2}\right) = \frac{1}{3!}\left(\tx{n} + \frac{1}{2}\right) +\frac{\z{n}}{3! \gamma} \nonumber\\
    & 
    = \sum_{k=1}^{\infty}  
    \frac{U_{kn}+1}{(k+3)!} \pm   \sum_{k=\ell}^{\infty} \frac{\z[k]{n}}{(k+3)!} \nonumber\\ 
    & = \sum_{k=1}^{\ell-2} \frac{U_{kn}\!+\!1}{(k+3)!}\!+\! 
    \frac{U_{(\ell-1)n}\!+\!1}{(\ell+2)!}\!+\! \sum_{k=\ell}^{\infty} \frac{U_{kn}\!+\!1 \!\pm\! \z[k]{n}}{(k+3)!}.
\end{align}
Note that $\sum_{k=1}^{\ell-2} \frac{U_{kn}+1}{(k+3)!}$ acts as the \emph{effective signal} for the 
$n$th channel use, while ${\w{n}=\frac{U_{(\ell-1)n}+1}{(\ell+2)!}+ \sum_{k=\ell}^{\infty} \frac{U_{kn}+1 \pm \z[k]{n}}{(k+3)!}}$ is its \emph{effective noise}. In the following, we show that the effective signal and the effective noise can be separated. 

Recall that $0\leq \z[k]{n} \leq k+2$ and $0 \leq U_{kn} \leq k$. Thus, $-k-1\leq U_{kn}+1 \pm \z[k]{n} \leq 2k+3$.
Therefore, we have
\begin{align*}
        \sum_{k=\ell}^{\infty} \frac{U_{kn}+1 - \z[k]{n}}{(k+3)!} & \geq   -\sum_{k=\ell}^{\infty} \frac{k+1}{(k+3)!} \stackrel{\rm (a)}{>} -\frac{1}{(\ell+2)!},
\end{align*}
where ${\rm (a)}$ follows from  Lemma~\ref{lm:sum-fact}. This further implies 
\begin{align}
\w{n}\geq \frac{U_{(\ell-1)n}+1}{(\ell+2)!}+ \sum_{k=\ell}^{\infty} \frac{U_{kn}+1 - \z[k]{n}}{(k+3)!} > 0. 
\end{align}

On the other hand,
\begin{align*}
    \sum_{k=\ell}^{\infty} &\hspace{-1pt}\frac{U_{kn}\!\hspace{-1pt}+1\hspace{-1pt}\! +\! \z[k]{n}}{(k+3)!}  \!\leq  \!\sum_{k=\ell}^{\infty} \frac{2k\!+\!3}{(k\!+\!3)!} \!<\! 2\hspace{-1pt} \sum_{k=\ell}^{\infty} \hspace{-1pt}\frac{k+2}{(k+3)!}\!
     \stackrel{\rm (a)}{=} \!\frac{2}{(\ell\!+\!2)!}.
\end{align*}
where ${\rm (a)}$ is followed from Lemma~\ref{lm:sum-fact}. Therefore, 
using the fact that $U_{(\ell-1),n} \leq \ell-1$, we arrive at  
\begin{align}
\w{n}&\leq \frac{U_{(\ell-1),n}+1}{(\ell+2)!}+ \sum_{k=\ell}^{\infty} \frac{U_{kn}+1 + \z[k]{n}}{(k+3)!} \nonumber\\
&< \frac{\ell-1}{(\ell+2)!} + \frac{2}{(\ell+2)!}  = \frac{\ell+1}{(\ell+2)!} < \frac{1}{(\ell+1)!}.
\end{align}
This together with Lemma~\ref{lm:exp-small} implies that $\w{n}$ can be written as 
$\w{n}=\sum_{k=\ell+1}^\infty\frac{\w[k]{n}}{(k+1)!} = \sum_{k=\ell-1}^\infty\frac{\w[k+2]{n}}{(k+3)!} $, where $\w[k]{n}\in\{0,\dots, k\}$. Thus, plugging this into~\eqref{eq:output-exp}, we get 
\begin{align}
    \sum_{k=1}^{\infty}  
    \frac{  \ty[k]{n}}{(k+3)!} = \sum_{k=1}^{\ell-2}  
    \frac{  \u[k]{n}+1}{(k+3)!} + \sum_{k=\ell-1}^\infty\frac{\w[k+2]{n}}{(k+3)!},
\end{align}
which implies 
\begin{align}
    \hu[k]{n} = \ty[k]{n} -1 = \left\{\begin{array}{ll}
         \u[k]{n}& \text{if $k\leq \ell-2$}  \\
         \w[k+2]{n} -1 & \text{if $k\geq \ell-1$.}  
    \end{array}
    \right.
\end{align}
Consequently, we have %$\hu[k]{n} \in\{0,1\dots, k\}$ for $k\leq \ell-2$ and 
$\hu[k]{n} \in\{-1,\dots, k+1\}$.
\end{proof}

\begin{prop}\label{prop:all-good}
If $\frac{|Z_n|}{6\gamma} <  \frac{1}{(\ell+2)!}$, for all $n \in[N]$,
then,
\begin{align}
    |\hat{U}-U| \leq \frac{2}{((\ell-1)!)^{N}}.
\end{align}
\end{prop}
\begin{proof}
Recall from Proposition~\ref{prop:each-good} that under the current noise regime, we have  $\hu[k]{n}=\u[k]{n}$ for $k=1,\dots, \ell-2$ and every $n\in [N]$, we have
\begin{align}
    |\hat{U}&\!-\!U| \!=\!\left|\sum_{k=1}^\infty \sum_{n=1}^N \frac{\hu[k]{n}\hspace{-1pt}-\hspace{-1pt}\u[k]{n}}{(k!)^N (k\!+\!1)^n}\right|\!=\!\left|\sum_{k=\ell-1}^\infty \sum_{n=1}^N \frac{\hu[k]{n}\hspace{-1pt}-\hspace{-1pt}\u[k]{n}}{(k!)^N (k\!+\!1)^n}\right|\nonumber\\ 
    &\leq \sum_{k=\ell-1}^\infty \sum_{n=1}^N \frac{\left| \hu[k]{n}-\u[k]{n}\right|}{(k!)^N (k+1)^n} 
    \stackrel{\rm (a)}{\leq} \sum_{k=\ell-1}^\infty \sum_{n=1}^N \frac{k+1}{(k!)^N (k+1)^n}\nonumber\\
    &=\sum_{k=\ell-1}^\infty \frac{k+1}{k}\frac{(k+1)^N-1}{((k+1)!)^N}
    \stackrel{\rm (b)}{\leq} 2\sum_{k=\ell-1}^\infty \frac{(k+1)^N-1}{((k+1)!)^N}\nonumber\\
    &= 2\left[\sum_{k=\ell-1}^\infty \frac{1}{(k!)^N}- \sum_{k=\ell}^\infty \frac{ -1}{(k!)^N}\right] = \frac{2}{((\ell-1)!)^N},
    \end{align}
    where ${\rm (a)}$ holds since $\hu[k]{n}\in\{-1,\dots, k+1\}$ and ${\u[k]{n}\in \{0,\dots, k\}}$, and ${\rm (b)}$ holds since $\frac{k+1}{k}\leq 2$. 
\end{proof}

Let us define the event $\mathcal{A}:=\{ |Z_n| <  \frac{6\gamma}{(\ell+2)!}, \forall n \in [N]\}$. Using Propositions~\ref{prop:bad} and~\ref{prop:all-good} and the union bound, we have

\begin{align}\label{eq:D:1}
    D & = \E[(\hat{U}-U)^2 | \mathcal{A}] \cdot  \P(\mathcal{A})+  \E[(\hat{U}-U)^2 | \mathcal{A}^c] \cdot \P(\mathcal{A}^c) \nonumber\\
    &  \leq \frac{4}{ ((\ell-1)!)^{2N}} + N \exp\left(-\frac{(\log (6\gamma/\sigma))^2}{2(\log \log (6\gamma/\sigma))^2}\right).
    \end{align}
From Sterling's approximation, we have ${n\hspace{-2pt}\leq\hspace{-1pt} n\log \frac{n}{e}\hspace{-1pt}\leq\hspace{-1pt} \log n!}$. This together with the relationship in~\eqref{eq:ell-range} implies
\begin{align}\label{eq:fac-bound}
    \ell+3 \leq  \log \ (\ell+3)! \leq \log  (6\gamma)/\sigma .
\end{align}
Therefore, we have
\begin{align}\label{eq:D-term1}
    \frac{1}{(\ell\!-\!1)!}&\!=\! \frac{\prod\nolimits_{i=0}^4 (\ell\!+\!i)}{(\ell+4)!} 
    \stackrel{\rm (a)}{\leq}
    \frac{\left(\ell\!+\!3\right)^{5} }{(\ell\!+\!4)!} \stackrel{\rm (b)}{\leq}
    \frac{\sigma}{6\gamma} \left(\log \frac{6\gamma}{\sigma}\right)^{5}\!\!,
\end{align}
where ${\rm (a)}$ follow from $ab \leq ((a+b)/2)^2$ and ${\rm (b)}$ is a consequence of~\eqref{eq:ell-range} and~\eqref{eq:fac-bound}. Plugging~\eqref{eq:D-term1} in~\eqref{eq:D:1}, we get 
\begin{align}\label{eq:D:2}
    D \hspace{-1pt}\leq \hspace{-1pt}4 \left(\frac{\sigma}{6\gamma} \left(\log \frac{6\gamma}{\sigma}\right)^{5}\right)^{2N}
    \!\!\!\!+  N \exp\left(-\frac{(\log (6\gamma/\sigma))^2}{2(\log \log (6\gamma/\sigma))^2}\hspace{-1pt}\right).
\end{align}
Plugging $\SNR=1/\sigma^2$, and setting 
\[
c_1=\frac{4}{(6\gamma)^{2N}2^{2N}}, \qquad c_2 = c_1(2\log 6\gamma)^{10N},
\]
and $c_3$ being the last term in~\eqref{eq:D:2}, we arrive at the distortion claimed in Theorem~\ref{thm:main}.  Note that at the high $\SNR$ regime (i.e., sufficiently small $\sigma$), the second term in~\eqref{eq:D:2} is negligible compared to the first term. 
Therefore, for ${\SDR\!=\!\frac{\E[|U-\E[U]|^2]}{D}= \frac{12}{D}}$ (for $U\sim \mathsf{Unif}([-1/2,1/2])$), we get 
\begin{align*}
    \log \SDR \geq N \log \SNR -10 N \log \log \SNR + o(\log\log \SNR).
\end{align*}
This completes the proof of Corollary~\ref{cor:SDR}. 

\newpage 
\begingroup
\let\cleardoublepage\clearpage
%\bibliography{Refs}

 \bibliographystyle{IEEEtran}
\endgroup

\newpage
\appendices

\clearpage

\appendix[Proof of the lemmas]
\label{expansionProof}

\begin{proof}[Proof of Lemma~\ref{lm:unique}]
Assume that $\{\{x_{ks}\}_{k=1}^\infty\}_{s=1}^S$ and $\{\{y_{ks}\}_{k=1}^\infty\}_{s=1}^S$ are two distinct progressive expansions for some $x\in [0,1]$, that is, 
\begin{align}
x=\sum_{k=1}^{\infty}\sum_{s=1}^{S}  
    \frac{x_{ks}}{(k!)^S (k+1)^s} = \sum_{k=1}^{\infty}\sum_{s=1}^{S}  
    \frac{y_{ks}}{(k!)^S (k+1)^s}. 
\end{align}
Then,  we prove that there exists some pair $(\ell,t)$ such that 
\begin{itemize}
    \item $x_{ks}=y_{ks}$ for every $(k,s)\prec (\ell, t)$;
    \item $x_{\ell t}=y_{\ell t}+1$;
    \item and, $x_{ks}=0$ and $y_{ks}=k$ for every $(k,s)\succ (\ell,t)$.
\end{itemize}
These three conditions imply that $\{\{x_{ks}\}_{k=1}^\infty\}_{s=1}^S$ is a finite $S$-progressive expansion of $x$, and $\{\{y_{ks}\}_{k=1}^\infty\}_{s=1}^S$ is the corresponding infinite expansion. More importantly, if $x$ does not admit a finite $S$-progressive expansion (which happens almost surely for $x\sim \mathsf{Unif}([-1/2,1/2])$), the expansion will be unique.

Let $(\ell,t)$ be the first pair  where $x_{\ell t}$ and $y_{\ell t}$ are distinct. That means $x_{ks}=y_{ks}$ for every $(k,s)\prec (\ell,t)$. Without loss of generality, assume $x_{\ell t} > y_{\ell t}$. Therefore, we have 
\begin{align}\label{eq:pr:unq:0}
    0 &= x-x = \sum_{k=1}^{\infty} \sum_{s=1}^{S}  
    \frac{y_{ks}-x_{ks}}{(k!)^S (k+1)^s}\nonumber\\
    &= \frac{y_{\ell t}-x_{\ell t}}{(\ell!)^S (\ell+1)^t} + 
    \sum_{s=t+1}^{S}  
    \frac{y_{\ell s}-x_{\ell s}}{(\ell!)^S (\ell+1)^s} \nonumber\\
    &\phantom{=} + 
    \sum_{k=\ell+1}^{\infty} \sum_{s=1}^{S}  
    \frac{y_{ks}-x_{ks}}{(k!)^S (k+1)^s}. 
\end{align}
Now, since $x_{\ell t}>y_{\ell t}$, we have 
\begin{align}\label{eq:pr:unq:1}
     \frac{y_{\ell t}-x_{\ell t}}{(\ell!)^S (\ell+1)^t} \leq  -\frac{1}{(\ell!)^S (\ell+1)^t}
\end{align}
Moreover, for 
$x_{\ell s}, y_{\ell s}\in \{0,\dots, \ell\}$ we have 
\begin{align}\label{eq:pr:unq:2}
    \sum_{s=t+1}^{S}  &
    \frac{y_{\ell s}-x_{\ell s}}{(\ell!)^S (\ell+1)^s} 
    \leq 
    \frac{1}{(\ell!)^S}\sum_{s=t+1}^{S}  
    \frac{\ell }{ (\ell+1)^s} \nonumber\\
    &\hspace{10mm}=  \frac{1}{(\ell!)^S}\left[\sum_{s=t+1}^{S}  
    \frac{\ell+1 }{  (\ell+1)^s} -
    \sum_{s=t+1}^{S}  
    \frac{1 }{ (\ell+1)^s} \right]\nonumber\\
    &\hspace{10mm}=  \frac{1}{(\ell!)^S}\left[\sum_{s=t}^{S-1}  
    \frac{1 }{  (\ell+1)^s} -
    \sum_{s=t+1}^{S}  
    \frac{1 }{ (\ell+1)^s} \right]\nonumber\\
    &\hspace{10mm}= \frac{1}{(\ell!)^S}\left[
    \frac{1 }{  (\ell+1)^t} -
    \frac{1 }{ (\ell+1)^S} \right]\nonumber\\
    &\hspace{10mm}= \frac{1}{(\ell!)^S (\ell+1)^t} - \frac{1}{((\ell+1)!)^S}.
\end{align}
Similarly, since $x_{ks}, y_{k s}\in \{0,\dots, k\}$, we can write 
\begin{align}\label{eq:pr:unq:3}
\sum_{k=\ell+1}^{\infty} \sum_{s=1}^{S} & 
    \frac{y_{ks}-x_{ks}}{(k!)^S (k+1)^s} 
    \leq \sum_{k=\ell+1}^{\infty} \sum_{s=1}^{S}  
    \frac{k}{(k!)^S (k+1)^s} \nonumber\\
    &= \sum_{k=\ell+1}^{\infty} \sum_{s=1}^{S}  \left[
    \frac{k+1}{(k!)^S (k+1)^s}  - \frac{1}{(k!)^S (k+1)^s}  \right]\nonumber\\
    &= \sum_{k=\ell+1}^{\infty} \frac{1}{(k!)^S}\left[\sum_{s=0}^{S-1} \frac{1}{ (k+1)^s}  - \sum_{s=1}^{S} \frac{1} {(k+1)^s} \right]\nonumber\\
    &= \sum_{k=\ell+1}^{\infty} \frac{1}{(k!)^S}\left[1  - \frac{1}{  (k+1)^S} \right]\nonumber\\
    &= \sum_{k=\ell+1}^{\infty} \frac{1}{(k!)^S} - \sum_{k=\ell+1}^{\infty} \frac{1}{((k+1)!)^S}\nonumber\\
    &= \sum_{k=\ell+1}^{\infty} \frac{1}{(k!)^S} \!-\! \sum_{k=\ell+2}^{\infty} \frac{1}{(k!)^S} =  \frac{1}{((\ell+1)!)^S}
\end{align}
Therefore, plugging~\eqref{eq:pr:unq:1}-\eqref{eq:pr:unq:3} into~\eqref{eq:pr:unq:0}, we get
\begin{align}
    0 \leq &-\frac{1}{(\ell!)^S (\ell+1)^t} + \frac{1}{(\ell!)^S (\ell+1)^t} \nonumber\\
    &- \frac{1}{((\ell+1)!)^S} + \frac{1}{((\ell+1)!)^S} =0.
\end{align}
That means all the inequalities in~\eqref{eq:pr:unq:1}-\eqref{eq:pr:unq:3} should hold with equality, which implies $x_{\ell t}=y_{\ell t}+1$, and $x_{k s}=0$ and $y_{k s }=k$ for every $(k,s)\succ (\ell,t)$. This completes the proof. 
\end{proof}

\begin{proof}[Proof of Lemma~\ref{lm:sum-fact}]
The desired identity can be proved using the following chain of equalities:
\begin{align*}
    \sum_{k=\ell}^{\infty} \frac{k}{(k+1)!} &= \sum_{k=\ell}^{\infty} \frac{k+1}{(k+1)!} - \sum_{k=\ell}^{\infty} \frac{1}{(k+1)!}  \nonumber\\
    &= \sum_{k=\ell}^{\infty} \frac{1}{k!} - \sum_{k=\ell}^{\infty} \frac{1}{(k+1)!}  \nonumber\\
    &= \sum_{k=\ell}^{\infty} \frac{1}{k!} - \sum_{k=\ell+1}^{\infty} \frac{1}{k!} =\frac{1}{\ell!}.
\end{align*}
\end{proof}

\begin{proof}[Proof of Lemma~\ref{lm:exp-small}] We prove the lemma by contradiction. Let the claim is wrong, and there exists some $(k_0,s_0)\prec (\ell,t)$ with $x_{k_0s_0}\geq 1$. Then,  we have
\begin{align*}
    \frac{1}{(\ell !)^S(\ell+1)^t}&>x=\sum_{k=1}^{\infty}\sum_{s=1}^{S}  
    \frac{x_{ks}}{(k!)^S (k+1)^{s}}  \nonumber\\
    &\geq \frac{x_{k_0s_0}}{(k_0!)^S (k_0+1)^{s_0}}\geq \frac{1}{(k_0!)^S (k_0+1)^{s_0}}.
\end{align*}
Therefore, 
\[
(k_0!)^S (k_0+1)^{s_0} >(\ell !)^S(\ell+1)^t,
\]
or equivalently, $(k_0, s_0) \succcurlyeq (\ell,t)$, which is in contradiction with assumption that $(k_0,s_0)\prec (\ell,t)$.     
\end{proof}

\begin{proof}[Proof of Lemma~\ref{lm:exp-prop}]
Let 
\begin{align}
X=\sum_{k=1}^{\infty}\sum_{s=1}^{S}   \frac{X_{ks}}{(k!)^S (k+1)^s},
\end{align}
be the $S$-progressive expansion of $X$. For $k\in \mathbb{N}$, $s\in \{1,\dots, S\}$, $n\in \{0,\dots, k\}$ and $i\in \{0,\dots, (k!)^S (k+1)^{s-1}-1\}$, define 
\[
\cA(i;k,s,n) := \left[\frac{(k+1)i+n}{(k!)^S (k+1)^{s}}, \frac{(k+1)i+n+1}{(k!)^S (k+1)^{s}}\right).
\]
First note that (see Figure~\ref{fig:expansion})
\begin{align}\label{eq:pe-intervals}
    \{X\in [0,1]:X_{ks}=n\} =\!\!\! \bigcup_{i=0}^{(k!)^S (k+1)^{s-1}-1} \cA(i;k,s,n).
\end{align}
Therefore, for $X$ uniformly distributed over $[0,1]$, we have 
\begin{align*}
\P(X_{ks} =n) &= \P\left[X\in \bigcup_{i=0}^{(k!)^S (k+1)^{s-1}-1} \cA(i;k,s,n) \right]\\
&= \sum_{j=0}^{(k!)^S (k+1)^{s-1}-1}|\cA(i;k,s,n)|\nonumber\\
&= \sum_{i=0}^{(k!)^S (k+1)^{s-1}-1} \frac{1}{(k!)^S (k+1)^{s}} \\
&= \frac{(k!)^S (k+1)^{s-1}}{(k!)^S (k+1)^{s}} = \frac{1}{k+1},
\end{align*}
for every $n\in\{0,1,\dots, k\}$. Therefore, each $X_{ks}$ admits a uniform distribution over $\{0,1,\dots, k\}$.

Now, consider $U_{ks}$ and $U_{\ell t}$, and without loss of generality, assume $k\geq \ell$. Then, similar to~\eqref{eq:pe-intervals}, we have 
\begin{align}\label{eq:pe-intervals:2}
\{X\in[0,1]:X_{\ell t} =m\} = \!\!\!\bigcup_{j=0}^{(\ell!)^S (\ell+1)^{s-1}-1} \cA(j;\ell,t,m).
\end{align}
Then, from~\eqref{eq:pe-intervals} and~\eqref{eq:pe-intervals:2}, we have
\begin{align*}
\{&X:X_{ks} =n, X_{\ell t}=m\} \nonumber\\
&= \!\left(\!\!\!\!\!\!\!\!\bigcup_{i=0}^{\ \  \ \ (k!)^S (k+1)^{s-1}-1} \!\!\!\!\!\!\!\!\cA(i;k,s,n)\!\right) \!\cap\! \left(\!\!\!\!\!\!\!\!\!\!\bigcup_{j=0}^{\ \ \ \ \ (\ell!)^S (\ell+1)^{t-1}-1} \!\!\!\!\!\!\!\!\!\cA(j;\ell,t,m)\!\right) \\
&= \bigcup_{i,j} \big(\cA(i;k,s,n) \cap \cA(j;\ell,t,m) \big).
\end{align*}
Therefore, 
\begin{align}
&\P(X_{ks}\!=\!n, X_{\ell t}\!=\!m) 
\!= \left|\bigcup_{i,j} \big(\cA(i;k,s,n) \cap \cA(j;\ell,t,m) \big) \right|\nonumber\\
&= \!\!\!\!\!\sum_{j=0}^{(\ell !)^S(\ell+1)^{t-1} -1} \sum_{i=0}^{(k!)^S (k+1)^{s-1}-1} \!\!\!\!\!\left|\cA(i;k,s,n) \cap \cA(j;\ell,t,m) \right|,\label{eq:P-joint}
\end{align}
where the last equality holds since the collection of intervals $\{\cA(i;k,s,n)\}_i$ are disjoint, and the collection of intervals $\{\cA(j;\ell,t,m)\}_j$ are disjoint. 
Without loss of generality, assume 
$(k,s) \succ (\ell,t)$. This implies $r:= \frac{(k!)^S (k+1)^s}{(\ell !)^S (\ell+1)^t}$ is integer and a multiple of $k+1$. Moreover, we have 
\begin{align}
\cA(j;\ell,t,m) &= \left[\frac{(\ell+1)j+m}{(\ell!)^S (\ell+1)^{t}}, \frac{(\ell+1)j+m+1}{(\ell!)^S (\ell+1)^{t}}\right)\nonumber\\
&= \left[\frac{r((\ell+1)j+m)}{(k!)^S (k+1)^{s}}, \frac{r((\ell+1)j+m+1)}{(k!)^S (k+1)^{s}}\right).
\end{align}
Note for intervals $[a_1,b_1)$ and $[a_2,b_2)$ we have 
\begin{align*}
    |[a_1,b_1) \cap [a_2,b_2)|  = (\min(b_1,b_2) - \max(a_1,a_2))^+.
\end{align*}
This implies that, for intervals with integer boundaries, we get
\begin{align}
    &|\cA(i;k,s,n) \cap \cA(j;\ell,t,m)| \nonumber\\
    &= \frac{1}{(k!)^S (k+1)^{s}} \Bigg|\Big[(k+1)i+n,(k+1)i+n+1\Big) \nonumber\\
    & \hspace{25mm} \cap \Big[r((\ell+1)j+m), r((\ell+1)j+m+1\Big)\Bigg|\nonumber\\
    &\!=\!\left\{
    \begin{array}{ll}
    \!\!\frac{1}{(k!)^S (k+1)^{s}} &  
    \begin{array}{c} \text{if } r(\ell+1)j+rm \leq (k+1)i+n  \\
    <r(\ell+1)j+rm +r,
    \end{array}\\
      \!\!   0 & \text{otherwise,}
    \end{array}
    \right.\nonumber\\
    &\!=\!\left\{
    \begin{array}{ll}
    \!\!\!\!\frac{1}{(k!)^S (k+1)^{s}} &  \text{if }\frac{r(\ell+1)j+rm-n}{k+1} \!\leq\! i \!<\hspace{-2.3pt} \frac{r(\ell+1)j+rm-n}{k+1}\!+\! \frac{r}{k+1}\\
      \!\!   0 & \text{otherwise}.\label{eq:int-intersection}
    \end{array}
    \right.
\end{align}
Since $i,\frac{r}{k+1}\in \mathbb{N}$, for each $j$, the condition in~\eqref{eq:int-intersection} holds for exactly $\frac{r}{k+1}$ values of~$i$. Incorporating this into~\eqref{eq:P-joint}, we get 
\begin{align*}
\P(&X_{ks}=n, X_{\ell t}=m) \nonumber\\
&= \sum_{j=0}^{(\ell !)^S(\ell+1)^{t-1} -1} \frac{r}{k+1} \times \frac{1}{(k!)^S (k+1)^{s}} \\
&= (\ell !)^S(\ell+1)^{t-1} \times\frac{\frac{(k!)^S (k+1)^s}{(\ell !)^S(\ell+1)^t}}{k+1} \times \frac{1}{(k!)^S (k+1)^{s}} \\
&= \frac{1}{(k+1)(\ell+1)} = \P(X_{ks}=n) \cdot \P(X_{\ell t}=m).
\end{align*}
This shows that $X_{ks}$ and $X_{\ell t}$ are pairwise independent. 
\end{proof}

\begin{proof}[Proof of Lemma~\ref{lm:fact-inverse-bound}]
We  define the function ${f(\omega):=\frac{\log \omega}{\log \log \omega}}$ for ${\omega\geq 4}$. Then, we have 
\begin{align*}
    f'(\omega) = \frac{(\log \log \omega -1)\log e}{\omega (\log \log \omega)^2}\geq 0.
\end{align*}
Hence, $f(\omega)$ is an increasing function for $\omega\geq 4$. Therefore, it suffices to show that for $\omega_n = n!$, we have $f(\omega_n)\leq n-1$. Note that it then implies that $f(\omega) \leq f(\omega_n)\leq n-1$ for  $4\leq \omega\leq \omega_n = n!$.

We have 
    \begin{align}\label{eq:lnF<:1}
        \log \omega_n &=\log n! =\sum_{m=1}^n \log m \leq \int_{1}^{n+1} \log x dx \nonumber\\
        &= \left(-x \log e + x \log x\right)\Big|_{x=1}^{x=n+1} \nonumber\\
        &= (n+1)\log(n+1) - n \log e. 
    \end{align} 
    Next, note that for $n\geq 6$ we have $e^{n-2}>(n+1)^2$. This implies
\begin{align}\label{eq:lnF<:2}
    (n&+1)\log (n+1) - n \log e - (n-1)\log (n-1) \nonumber\\
    &= \log \frac{(n+1)^{n+1}}{(n-1)^{(n\!-\!1)} e^n} = \log \left(\!1\!+\!\frac{2}{n\!-\!1}\right)^{n\!-\!1} \!\!\!+ \log \frac{(n\!+\!1)^2}{e^{n}} \nonumber\\
    &\leq \log e^2 + \log \frac{n+1}{e^{n}} = \log \frac{(n+1)^2}{e^{n-2}} \leq \log 1 =0.
\end{align}
Using~\eqref{eq:lnF<:2} in~\eqref{eq:lnF<:1}, we get 
\begin{align}\label{eq:lnF<:3}
    \log \omega_n \leq (n-1)\log (n-1).
\end{align}
Similarly, we can write 
\begin{align}\label{eq:lnF>:1} 
        \log \omega_n &= \log n! = \sum_{m=1}^n \log m \geq  \int_{1}^{n} \log x dx \nonumber\\
        &= \left(-x \log e \!+\! x \log x\right)\Big|_{x=1}^{x=n} = n \log n -(n\!-\!1) \log e \nonumber\\
        &\geq  (n-1) \log \frac{n}{e}.
    \end{align}
Taking $\log(\cdot)$ from both sides of~\eqref{eq:lnF>:1}, we arrive at
\begin{align}\label{eq:lnF>:2} 
    \log \log \omega_n \geq \log (n-1) + \log \log \frac{n}{e} \geq \log (n-1),
\end{align}
where the last inequality holds for $n\geq 2e$. Dividing~\eqref{eq:lnF<:3} by~\eqref{eq:lnF>:2}, we get  
    \begin{align*}
        n-1\geq  \frac{\log \omega_n}{\log \log \omega_n} = f(\omega_n),
    \end{align*}
    which completes the proof. 
\end{proof}

\begin{proof}[Proof of Lemma~\ref{lm:input-properties}]
    
Since $x_{kn}\in\{0,1,\dots, k\}$, we have 
\begin{align*}
    \tx{n} &=3!\sum_{k=1}^{\infty}  
    \frac{\u[k]{n}+1}{(k+3)!} - \frac{1}{2} \leq 3!\sum_{k=1}^{\infty}  
    \frac{k+1}{(k+3)!} - \frac{1}{2} \nonumber\\
    &= 6\sum_{k=1}^{\infty}  
    \frac{(k+3)-2}{(k+3)!} - \frac{1}{2} \\
    &=6 \left[\sum_{k=1}^{\infty}  
    \frac{1}{(k+2)!} -2\sum_{k=1}^{\infty}  
    \frac{1}{(k+3)!}\right] -\frac{1}{2}\nonumber\\
    &=
    6\left[\left(e-\frac{5}{2}\right)-2 \left(e-\frac{8}{3}\right)\right]- \frac{1}{2} \\
    &=6\left(\frac{17}{6}-e\right)-\frac{1}{2} = 16.5-6e   =0.1903.
\end{align*}
Similarly, we can write
\begin{align*}
     \tx{n} &= 3!\sum_{k=1}^{\infty}  
    \frac{\u[k]{n}+1}{(k+3)!} - \frac{1}{2} \geq   3!\sum_{k=1}^{\infty}  
    \frac{1}{(k+3)!} - \frac{1}{2} \nonumber\\
    &= 
    6\left[e-\frac{8}{3}\right]- \frac{1}{2} 
    =6e- 16.5 = -0.1903.
\end{align*}

For the first moment of $\tx{n}$ we have
\begin{align*}
    \E\left[\sum_{k=1}^\infty \frac{U_{kn}+1}{(k+3)!}\right] &=
  \sum_{k=1}^\infty \frac{1}{(k+3)!}  \E[U_{kn}+1] \nonumber\\
  &\stackrel{\rm (a)}{=}  \sum_{k=1}^\infty \frac{1}{(k+3)!} \frac{k+2}{2} \stackrel{\rm (b)}{=} \frac{1}{2}\times \frac{1}{6} = \frac{1}{12}. 
\end{align*}
Note that ${\rm (a)}$ holds for $U_{kn}$ with uniform distribution over $\{0,1\dots, k\}$, and ${\rm (b)}$ follows from Lemma~\ref{lm:sum-fact}. 
Therefore, we get
\begin{align*}
    \E\left[\tx{n}\right] =  \E\left[3!\sum_{k=1}^\infty \frac{U_{kn}+1}{(k+3)!}-\frac{1}{2}\right] &= 3! \frac{1}{12}- \frac{1}{2}=0. 
\end{align*}

Similarly, we have 
\begin{align*}
    \E&\left[\left(\sum_{k=1}^\infty \frac{U_{kn}+1}{(k+3)!}\right)^2\right] = \E\left[\sum_{k,\ell=1}^\infty \frac{(U_{kn}+1)(U_{\ell n}+1)}{(k+3)! (\ell+3)!}\right] \\
    &= \E\left[\sum_{k=1}^\infty \frac{(U_{kn}+1)^2}{((k+3)!)^2} + 2\sum_{1\leq k <\ell} \frac{(U_{kn}+1)(U_{\ell n}+1)}{(k+3)! (\ell+3)!}\right]\\
    &\stackrel{\rm (a)}{=} \sum_{k=1}^\infty \frac{1}{((k+3)!)^2} \E[(U_{kn}+1)^2] \nonumber\\
    &\phantom{=}+ 2\sum_{1\leq k <\ell} \frac{1}{(k+3)! (\ell+3)!} \E[U_{kn}+1]\E[U_{\ell n }+1]\\
    &\stackrel{\rm (b)}{=} \sum_{k=1}^\infty \frac{1}{((k+3)!)^2} \frac{(k+2)(2k+3)}{6} \nonumber\\
    &\phantom{=}+ 2 \sum_{1\leq k <\ell} \frac{1}{(k+3)! (\ell+3)!}  \frac{k+2}{2} \frac{\ell+2}{2}\\
    &\stackrel{\rm (c)}{=}  \sum_{k=1}^\infty \frac{1}{((k+3)!)^2} \frac{k(k+2)}{12}
    + 
    \sum_{k=1}^\infty \frac{1}{((k+3)!)^2} \frac{(k+2)^2}{4}\nonumber\\
    &\phantom{=}+ 2 \sum_{1\leq k <\ell} \frac{1}{(k+3)! (\ell+3)!}  \frac{k+2}{2} \frac{\ell+2}{2}\\
    &= \sum_{k=1}^\infty \frac{1}{((k+3)!)^2} \frac{k(k+2)}{12}
    +  \frac{1}{4}\left(\sum_{k=1}^\infty \frac{k+2}{(k+3)!}\right)^2\\
    &\stackrel{\rm (d)}{=}   \alpha + \frac{1}{4}\left(\frac{1}{6}\right)^2
\end{align*}
where $\alpha =\frac{1}{48}(16 I_0(2)-16 I_1(2)-11)\approx 0.000482816$, and $I_1(z)$ is the modified Bessel function of the first kind.  Note that the step in~${\rm (a)}$ holds since $U_{kn}$ and $U_{\ell n}$ are independent, ${\rm (b)}$ follows from the fact that $U_{kn}$ is uniformly distributed over $\{0,1,\dots, k\}$, the equality in~${\rm (c)}$ holds for $(k+2)(2k+3)/6 = k(k+2)/6 + (k+2)^2/4$, and ${\rm (d)}$ follows from Lemma~\ref{lm:sum-fact}. 

Therefore, from~\eqref{eq:channel-input} we get
\begin{align*}
    \E[|\tx{n}|^2] &\!=\!  36\E\hspace{-1pt}\left[\!\left(\sum_{k=1}^\infty \frac{U_{kn}+1}{(k+3)!}\right)^{\!\!2}\hspace{-1pt}\right] \!\!-\! 6 \E\hspace{-1pt}\left[ \sum_{k=1}^\infty \frac{U_{kn}+1}{(k+3)!}\right] \hspace{-1pt}\!+\! \frac{1}{4} \\
    &= 36 \alpha + \frac{1}{4} -\frac{1}{2} + \frac{1}{4} = 36\alpha = 0.0173814.
\end{align*}
This completes the proof of the lemma. 
\end{proof}

\end{document}